\documentclass[submission,copyright,creativecommons]{eptcs}
\providecommand{\event}{FMAS 2023} 

\usepackage[utf8]{inputenc}

\usepackage{iftex}

\ifpdf
  \usepackage{underscore}         
  \usepackage[T1]{fontenc}        
\else
  \usepackage{breakurl}           
\fi

\title{Enforcing Timing Properties in Motorway Traffic}
\author{Christopher Bischopink
\institute{Carl von Ossietzky University Oldenburg, Oldenburg}
\email{bischopink@informatik.uni-oldenburg.de}
}
\def\titlerunning{Enforcing Timing Properties}
\def\authorrunning{C. Bischopink}

\usepackage{graphics}
\usepackage{csquotes}
\usepackage{todonotes}
\usepackage{amssymb}
\usepackage{amsthm}
\newtheorem{theorem}{Theorem}
\newtheorem{example}{Example}
\newtheorem{lemma}{Lemma}
\newtheorem{remark}{Remark}

\usepackage{tikz}
\usetikzlibrary{arrows,shapes,automata,petri,backgrounds,calc,intersections,decorations.pathreplacing}

\usepackage{amsmath,MnSymbol}
\usepackage{amsfonts}
\usepackage{./macros}

\usepackage{wrapfig}

\usepackage{algorithm}
\usepackage{algpseudocode}
\begin{document}
\maketitle

\begin{abstract}
	In previous work \cite{BO23}, we proposed a \emph{Runtime Enforcement} Approach to deal with timing properties in motorway traffic, which are present in form of \emph{Timed Multi-Lane Spatial Logic (TMLSL)} formulae, a logic tailored to express both spatial and timing properties.
	Employing communication between the cars, we utilised a nondeterministic controller \enquote{guessing} which actions to execute next for each car, before asking the local monitors of the cars for permission to execute the announced actions.
	In this contribution, we consider a more reasonable controller that only considers sequences that satisfy its own properties. This is done utilising region automata that one can generate from the cars' specifications. In the approach, we also came along a minor decidability result for TMLSL.
\end{abstract}

\section{Introduction}
\label{sec:introduction}

With the number of (at least partially) autonomous cars increasing on the roads around the globe, challenges and advantages in the specification and verification of their behaviour occur. If one assumes that the cars are able to communicate with each other, a more detailed interplay between them is possible than with human drivers and allows finding solutions for complicated traffic situations that human drivers could easily miss. 

The roads we consider here are motorways, formalised as \emph{traffic snapshots} \cite{HLOR11} with a logic to reason about them called \emph{Multi-Lane Spatial Logic (MLSL)} \cite{HLOR11}. We extended this spatial logic towards \emph{Timed Multi-Lane Spatial Logic (TMLSL)} to also cover the timing aspect of a car's specification in \cite{BO22}. Based on TMLSL, we proposed a runtime-enforcement approach in \cite{BO23}, employing a nondeterministic controller that asked for the permission of other cars for the actions it wants to execute. Due to the nondeterminism, completely unreasonable sequences that even violate the own car's specification could be announced. In the lack of a result that allow announcing/checking only reasonable sequences, the nondeterministic controller still allowed us to show that the approach is complete.

In this work, we propose a more reasonable approach, utilising the region automaton of the cars' specifications. Still, all cars announce sequences they want to achieve, but this time all announced sequences would yield satisfying runs at least for the car that announces them. The announced sequences are then checked by a central entity, e.g. a road-side unit, for a run that is satisfying for all cars and informs the cars accordingly. We furthermore present a minor decidability result, eliminating one of the causes for the semi-decidability of TMLSL \cite{BO22}.

\paragraph{Related Work}
In the context of MLSL, different topologies have been explored in addition to motorway traffic, namely country roads \cite{HLO13} and urban traffic \cite{Sch18}, as well as their satisfaction problems \cite{Ody15b}\cite{FHO15} and controllers for cars in these topologies with different desirable properties such as liveness and fairness \cite{BS19,Sch18b}. Other approaches in the context of autonomous or automated driving systems use e.g. \emph{differential dynamic logic} \cite{LPN11} or a specification with extended types of timed automata \cite{LMT15}. These approaches mostly concentrate on a top-level view of the system under control. A more technical view of the evolution of a cars dynamics in an adaptive cruise control setting is e.g. given in \cite{AMP21}.
Runtime Enforcement \cite{Sch00} and runtime verification \cite{H08} are also well studied topics. To the best of our knowledge, runtime enforcement approaches however are more intensively considered in  more restricted settings than motorway traffic, where the system evolves quiet dynamically to the input given. Another feature is that in our case, the input and output of the system under control are different from each other.
\section{Preliminaries}
\label{sec:prelim}

In this Section, we introduce the formal concepts our approach is build on. We start with the model of motorway traffic, its logic and evolution in Sect.~\ref{ssec:MLSL} and continue with the timing model used, Sect.\ref{ssec:time}. The combination of them is called TMLSL and covered in Sect.~\ref{ssec:TMLSL}.

\subsection{Spatial Model of and Logic for Motorway Traffic}
\label{ssec:MLSL}

\paragraph{Model}
The spatial model we use in the setting of motorway traffic was introduced in \cite{HLOR11}. It allows only traffic in one direction on a fixed set of lanes $ \mathbb{L}=\{1,\dots, N\} $ with an infinite extension each. On these lanes, cars from the set of car identifiers $ \mathbb{I}=\{A,B,\dots\} $ drive, each car $ C $ of them with a certain speed $ \spd(C) $, acceleration $ \acc(C) $ and position on the lane, $ \pos(C) $. There are two different types of occupation a car can have on lanes, either a reservation $ \res: \mathbb{I} \rightarrow \mathcal{P}(\mathbb{L}) $, the space it physically occupies (multiple lanes if it is changing lanes at this moment) or a claim $ \clm: \mathbb{I} \rightarrow \mathcal{P}(\mathbb{L}) $, the lane a car wishes to change to, which is the equivalent of setting a turn signal. Altogether, this information is represented as a traffic snapshot $ \TS=(\res,\clm,\pos,\spd,\acc) $. 

In a traffic snapshot, there is no information present what the sizes of the cars and their braking distances are, as $ \pos $ only stores the rear end of each vehicle. We neglect the concept of a \emph{Sensor Function} here that makes this information available to us and simply assume that size and braking distance of each car is known. Also omitted is the \emph{View} that allows to only consider a finite extension of the infinite extension of a traffic snapshot when evaluating formulae. A graphical representation of three traffic snapshots is depicted in Fig.~\ref{fig:examplesequence}, where we also omitted showing concrete values for the position, speed and acceleration of the cars.

As already hinted at, a traffic snapshot describes the situation on the road at one point in time only. A situation on the road may evolve, which is handled in the model via \emph{transitions}.

\paragraph{Transitions}

We divide the set of transitions usable in a traffic snapshot into transitions regarding the discrete behaviour between lanes and transitions regarding the continuous behaviour along the lanes. The first set consists of car $ C $ claiming a lane $ n $ resp. withdrawing all claims ($ \claim{C,n}/\wdc{C} $) and car $ C $ reserving a lane resp. withdrawing all reservations except the one on lane $ n $ ($ \r{C}/\wdr{C,n} $). The second set is the one we focus more on, as it is considered more intensively in what follows, they handle the change of a car's acceleration to some value $ a $ ($ \mathsf{acc}(C,a) $) as well as the passing of $ t $ time units ($ t $). In the following definition, $ \oplus $ is the overriding operator of Z~\cite{WD96}:
\begin{align*}
	&\mathcal{TS} \xrightarrow{t} \mathcal{TS'} &\Leftrightarrow\ \ & \mathcal{TS'}=(\res,\clm,\pos',\spd',\acc)\\
	&&&\land \forall C \in \mathbb{I}:\pos'(C)=\pos(C)+\spd(C)\cdot t+\frac{1}{2}\acc(C)\cdot t^2\\
	&&&\land \forall C \in \mathbb{I}: \spd'(C)=\spd(C)+\acc(C)\cdot t \\
	&\mathcal{TS} \xrightarrow{\mathsf{acc}(C,a)} \mathcal{TS'}& \Leftrightarrow\ \  &\mathcal{TS'}=(\res,\clm,\pos,\spd,\acc')\\
	&&&\land \acc' =\acc\oplus\{C\mapsto a\},
\end{align*}

Over the set of actions, which are the transitions without the one where only time passes, we define \emph{timed words of actions} $ \omega = \langle (\alpha_0,t_0),\dots,(\alpha_n,t_n) \rangle $, with $ \alpha_i $ an action and $ \langle t_0,\dots,t_n\rangle $ a real-time sequence. For a formal account, we refer to \cite{Ody20}.

A graphical representation of a transition sequence including two transitions (one discrete and one continuous) is shown in Fig.~\ref{fig:examplesequence}. It can be interpreted as the timed word $ \omega_1=\langle(\r{A},3)\rangle $.

\begin{remark}[Dynamic behaviour of the cars]
	The model used for describing the dynamics of the cars is a quiet simple one, ignoring many difficulties that one would encounter in the real world, such as friction or variable acceleration capabilities based on the current speed. Still, we believe that it is a good abstraction of the real world's dynamics. Especially if one considers that the positions (plus size and braking distance) could be over-approximations, this allows some degree of freedom in achieving a behaviour to match the correct positions.
\end{remark}

\begin{figure}
	\includegraphics[width=\linewidth]{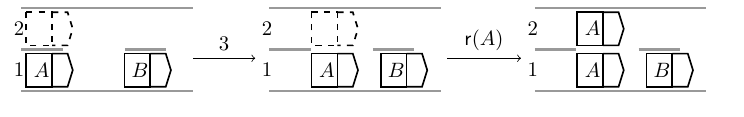}
	\caption{A transition sequence with three traffic snapshots and two actions, taking place with two cars on two lanes. Car $ A $ is faster than car $ B $, so it comes closer when $ t=3 $ time units pass. Afterwards, it reserves its formerly claimed space (dashed copy of it on the neighbouring lane), one step further in an attempt to overtake car $ B $.}
	\label{fig:examplesequence}
\end{figure}

\paragraph{Logic}

To reason about the traffic situations formalised as traffic snapshots, the logic MLSL \cite{HLOR11} was introduced . Here, we consider a variant of MLSL called MLSLS (MLSL \emph{with scopes}) \cite{FHO15}, limiting the range of cars over which car variables are evaluated to a finite range. Formulae of MLSLS are constructed according to the grammar 

$$\varphi ::= \gamma=\gamma' \mid  \mathit{free}  \mid  \mathit{re}(\gamma) \mid  \mathit{cl}(\gamma) \mid  l=k \mid \exists c.\varphi \mid  {\varphi_1}\hchop{\varphi_2} \mid  \vchop{\varphi_1}{\varphi_2} \mid  cs:\varphi,$$
and standard Boolean combinations of such formulae, where $ \gamma$ and $\gamma'$ are car variables, $ k\in\mathbb{R} $ and $ cs $ is a (sub-)set of car variables.

The atoms that formulae are constructed from are the comparison of two car variables $ \gamma = \gamma' $, $ \free $ denotes a segment with free space that is not occupied by others, the reservation $ \re{\gamma} $ of a car $ \gamma $, the claim $ \cl{\gamma} $ of a car $ \gamma $ and the comparison $ l=k $ of the length $ l $ of the considered segment against some value $ k $. With $ \exists c.\varphi $ one asks for the existence of a car $ c $ that satisfies $ \varphi $. The horizontal chop operator $ \left( {\varphi_1}\hchop{\varphi_2}\right) $ is used to determine if it is possible to divide the current segment into two parts along the lanes s.t. in the first part $ \varphi_1 $ and in the part directly ahead of it $ \varphi_2 $ holds. The same can be specified with the vertical chop operator $ \left(\vchop{\varphi_1}{\varphi_2} \right)$, but the point to divide is in between two lanes this time. $ cs:\varphi $ limits the scope over which $ \varphi $ is evaluated to the finite domain $ \mathit{cs} $ and effects only formulae that use quantification over the cars in their semantics and thus only $ \exists c.\varphi $ and $ \free $. In Sect.~\ref{sec:decidability}, we focus on length comparisons $ l=k $ and simply abbreviate them as $ \theta $, as there can be multiple of them regarding the same segment.

A common abbreviation used is the \emph{somewhere} modality $ \somewhere{\varphi} $, expressing that there is a partition on the road along and in between the lanes s.t. $ \varphi $ holds in some point of the partition.

For a formal definition, especially about the exact semantics of MLSLS formulae, we refer the reader to \cite{FHO15}, but would like to point out that MLSLS is, in contrast to pure MLSL, decidable.

\subsection{Model and Logic of Time}
\label{ssec:time}
The logic we consider here for the timing aspects is called \emph{State-Clock Logic (SCL)} \cite{RS97}. Formulae from this logic are constructed over an alphabet of propositions $ \Sigma $ according to the grammar

$$ \psi ::= p \ \mid \ \psi_1 \lor \psi_2 \ \mid \  \lnot \psi \ \mid \ \until{\psi_1}{\psi_2}\ \mid \ \since{\psi_1}{\psi_2}\ \mid \ \rtnext{\sim c}{\psi}\ \mid \ \rtlast{\sim c}{\psi}, $$
with  $\ \sim\ \in \{<,\leq,=,\geq,>\} $ and $ p\in\Sigma $.

Apart from well-known Boolean combinations of formulae and the usual until- and since operators ($ \until{} $ resp. $ \since{} $), SCL allows to measure the time since ($ \vartriangleleft $)/ until ($ \vartriangleright $) a formula $ \psi $ held/holds for the last/next time and compare this difference with $ \sim c $. 

The semantics of SCL formulae is evaluated on (usually infinite) \emph{timed sequences of states} $ m=\langle(s_0,I_0,),(s_1,I_1),\dots \rangle $ with $ s_i\subseteq\Sigma $ and $ \langle I_0,I_1,\dots \rangle $ a monotonically increasing sequence of adjacent intervals. Intuitively, the formula $ \rtnext{\sim\ c}{\psi} $ holds at time point $ t $ in the $ i $th state of $ m $, written $ (m,i,t)\vDash\rtnext{\sim\ c}{\psi} $ iff there is some state $ (s_j,I_j) $ at position $ j>i $ where $ \psi $ holds, all states in between $ i $ and $ j $ do not satisfy $ \psi $ and the difference between the left border of the $ j$th interval $ I_j $ and $ t $ satisfies $ \sim\ c $ . The analogous applies for the operator $ \rtlast{\sim\ c}{\psi} $, the semantics of the remaining operators is as expected. Example for both syntax and semantics are given in the next section, for a more formal account on the topic we refer the reader to \cite{RS97}.

The decidability problem of SCL is known to be decidable. Given a SCL formula $ \psi $, one can construct a \emph{State-Clock (SC)} Automaton $ A_\psi=(\mathcal{P},C_\mathcal{P},L,L_0,E,\mathcal{L},\Delta,\mathcal{F}) $, with $ \mathcal{P} $ the set of propositions used, $ C_\mathcal{P} $ the set of clocks, $ L $ and $ L_0 $ the (initial) locations of $ A_\psi $, a transition relation $ E\subseteq L\times L $, a labelling function $ \mathcal{L} $ assigning the propositions valid in it to every location of $ L $, another labelling function $ \Delta $ assigning constraints over $ C_\mathcal{P} $ to every location of $ L $ and a family of B\"uchi acceptance sets $ \mathcal{F} $.  From there on a \emph{Region-Automaton} $ \mathcal{R}(A_\psi) $ \cite{RS97}\cite{AD94} can be constructed. Having operators to compare the time to for both the future and the past, SC automata have a history clock $ x_p $ and a prophecy clock $ y_p $ for each proposition $ p \in \mathcal{P}$ in $ C_\mathcal{P} $. Both types of clocks need to be respected when constructing the region automaton. A \emph{region} $ [\nu] $ describes a class of clock valuations $ \nu $ that cannot be distinguished  by any SC automaton and can be represented as a set of (in-)equalities over the set of clock variables $ x_y,\ y_p $  and the natural numbers $ \mathbb{N} $. Iff the language of $ \mathcal{R}(A_\psi) $ is not empty, the formula $ \psi $ is satisfiable. In \cite{BO23}, we extended State-Clock automata with \emph{broadcast communication} like the timed automata of UPPAAL \cite{UPPAAL97}. For a formal account on the broadcast communication used we refer to \cite{Sch18} and only point out that sending some data $ d $ over a channel $ c $ is denoted as $ c!\langle d\rangle $ and receiving this data on the same channel is denoted as $ c?\langle d\rangle $. We also allowed simple functions dealing with data structures and simple computations on the transitions of communicating SC automata.

\subsection{TMLSL}
\label{ssec:TMLSL}

To express and reason about both spatial and timing properties in motorway traffic, we introduced \emph{TMLSL} \cite{BO22}. The idea of this logic is to use MLSLS-formulae as the propositions that SCL formulae are build from.  The intuitive idea for the semantics is that a traffic snapshot $ \TS $ with a timed word of actions $ \omega  $ is a model for a formula $\varphi$, $ \TS_0,\omega\vDash \varphi $  iff there is a timed sequence of states $ m(\TS_0,\omega) $ that is propositionally consistent and complete in the subformulae of $ \varphi $, describes the evolution of $ \TS $ along $ \omega $ and is a model of $ \varphi $ in the SCL-semantics, $ m(\TS,\omega)\vDash_{SCL}\varphi $. We now give an example for a TMLSL formula as well as their satisfaction. For simplicity, we use car identifiers instead of car variables in the MLSLS propositions in the formula.

\begin{example}[TMLSL]
	\label{ex:tmlsl}
	Consider the TMLSL formula $$\varphi_1= \somewhere{\re{A}\hchop\free\land{l=21}\hchop\re{B}} \implies \rtnext{=5}{\somewhere{\re{A}\hchop\free\land{l=15}\hchop\re{B}}}$$ that specifies that when the distance between the reservations of the two cars $ A $ and $ B $ is equal to $ 21 $ distance units somewhere in the traffic snapshot, it needs to be equal to $ 15 $ distance units within exactly $ 5 $ time units. A satisfying sequence of states is 
	\begin{align*}
		 &m=\langle(\somewhere{\re{A}\hchop\free\land{l=21}\hchop\re{B}},[0,0]),(\somewhere{\re{A}\hchop\free\land{l\leq21\land l>15}\hchop\re{B}},(0,5)),\\&(\somewhere{\re{A}\hchop\free\land{l=15}\hchop\re{B}},[5,7))\rangle, 
	\end{align*}
	
	 one that is not is
	 
	 \begin{align*}
	 	&m'=\langle(\somewhere{\re{A}\hchop\free\land{l=21}\hchop\re{B}},[0,0]),(\somewhere{\re{A}\hchop\free\land{l\leq21\land l>15}\hchop\re{B}},(0,5]),\\&(\somewhere{\re{A}\hchop\free\land{l=15}\hchop\re{B}},(5,7))\rangle,
	 \end{align*}
  as the distance between the two cars reached the value $ 15 $ too late. The only difference between $ m $  and $ m' $ are the shapes of the second and third interval. Please note that we omitted some subformulae of $ \varphi_1 $ in the sequences in an attempt to keep them readable.
	
	Ex.~\ref{ex:running} gives values for the positions, speeds and accelerations of the cars in the first traffic snapshot $ \TS $ of Fig.~\ref{fig:examplesequence} and a timed word of actions $ \omega $ s.t. applying $ \omega $ to $ \TS $ results in the timed sequence of states $ m(\TS_0,\omega) $ that satisfies $ \varphi_1 $ written $ \TS,\omega\vDash\varphi_1 $, resp. $ m(\TS_0,\omega)\vDash_{SCL}\varphi_1 $.
\end{example}

\begin{remark}
	Using MLSLS formulae as the propositions of SCL rather than MLSL actions imposes some difficulties, still some situations on the road can only be described using formulae rather than action. Additionally, actions take zero time in the model, so one could argue that they are not observable from the outside. So for the cars on the motorway that we want to control, we have MLSL actions as the input to them, but the system produces evolutions of traffic snapshots as an output, which we observe through MLSLS formulae.
\end{remark}

\paragraph{Finite Semantics}
In \cite{BO23}, we introduced a finite semantics for TMLSL, suited for the finite sequences that are usually available in runtime monitoring/enforcement and the reality on motorways. Intuitively, a finite word of actions $ \omega $ satisfies a formula $ \varphi $ in the finite semantics up to time $ t $, $ \TS,\omega\vDash_t\varphi $ iff there exists at least one suffix $ m' $ s.t. $ m(\TS,\omega).m' \vDash_{SCL} \varphi$ in the infinite semantics.
\section{Decidability Results}
\label{sec:decidability}

Regarding the decidability of TMLSL, we point out that the logic is at least semi-decidable \cite{BO22}. In answering this question, we considered maximum values on the acceleration of the cars (from $ \acc_{min}$  to $ \acc_{max} $) and the speed (from $ 0 $ to $ \spd_{max} $), as in the real world there are (at least) physical bounds, too. We do the same here. The decidability results of SCL and MLSLS do not directly transfer to TMLSL, as we need to interpret the SCL-propositions, which are MLSLS-formulae, and the cars in the traffic snapshot may not be able to behave as specified in TMLSL. One cause for only semi-decidability are the actions regarding the dynamic behaviour along the lanes, the change of a car's acceleration and the passing of time. While it is easy to see what discrete actions need to be executed and when, given a timed sequence of states $ m $, it was unknown how many acceleration changes are needed to achieve cars behaving correct with respect to the lengths constraints specified in $ m $. In this paper, we give an algorithm that decides this question. Before doing so, however, we start with an example (adjusted example of \cite{BO22}):

\begin{example}
	\label{ex:running}
	Consider a traffic snapshot with one lane and two cars $ A $ and $ B $, $ B $ driving ahead of $ A $, where the distance between the two cars is equal to $ 21 $ and both of them have a speed of $ 4 $. For simplicity, we furthermore assume that car $ B $ cannot change its acceleration, it is fixed at $ 0 $, the initial acceleration of $ A $ does not matter. For this traffic situation, we have a specification expressing that the distance between the two cars needs to be equal to $ 15 $ within $ 5 $ time units, formalised as the formula $ \varphi_1 $ from Ex.~\ref{ex:tmlsl}. Furthermore, assume that we have $ \acc_{min}=-10,\ \acc_{max}=5 $ and $ \spd_{max}=13 $ as bounds on the dynamic behaviour.  In this example, there is no timed word of actions that allows the traffic snapshot to behave as specified, if we only allow acceleration changes at one point in time, as we either obtain a speed to fast or need to accelerate stronger than the specified bounds allow. If we allow acceleration changes at two points, there is a solution: $ \omega=\langle(\acc(A,0.75),0),(\acc(A,-6),4)\rangle $. Letting one further time unit pass results in a distance exactly $ 15 $.
\end{example}

As described in Sect.~\ref{ssec:MLSL}, the dynamics of each car $ C $ evolves according to the simple mechanical equation $ \pos'(C) = \pos(C) + \spd(C)\cdot t + \frac{1}{2}\cdot\acc(C)\cdot t^2 $, with $ t $ being the time that elapses and $ \pos'(C) $ the new position of $ C $. The speed evolves according to $ \spd'(C)=\spd(C)+t\cdot\acc(C) $.

For a finite timed sequence of states $ m=\langle(s_0,I_0),\dots,(s_m,I_m) \rangle  $ and a number $ n $ of points in which we can split the interval $ I= [0,t]=\langle I_0,\dots,I_m\rangle $, we define $ \mathit{DYN}(m,n,I)$ as the set of equations that describe the solution space of $ m $ on the (timing) interval $ I $ for these $ n $ splitting points. In the equations listed below, we only consider the length measurements/constraints that we need to satisfy in $ s_i $, as the question which discrete actions one need to execute between two phases is easy to answer. We summarise these constraints as $ \theta_m(t) $ for the length constraints that occur in the phase $ (s_i,I_i) $ of $ m $ with $ t\in I_i $. Iff the difference in the position of any two cars affected by it satisfies $ \theta_m(t) $ at point $ t $, we denote this as $ \Delta \pos(t) \vDash  \theta_m(t) $. Please note that $ \theta_m(t) $ can consist of an arbitrary number of constraint, e.g. when we require that the distance between two cars is smaller than some value and greater than some other value, for example when we want to exclude (potential) collisions while being quite close to the car in front. Additionally, $ \theta_m(t) $ can constrain the distance between more than two cars.

\begin{align}
	&\mathit{DYN}(m,n,I) =\notag \\
	 &\pos_0(C) \text{ and }  \spd_0(C) \text{ are as in } \TS_0 ,\\
	  &\pos_n(C) = \pos_{n-1}(C)+ \spd_{n-1}(C)\cdot t_{n-1} \cdot \frac{1}{2}\cdot \acc_{n-1}(C)\cdot t_{n-1}^2 ,\\
	  &\spd_n(C)= \spd_{n-1}(C) + \acc_{n-1}(C)\cdot t_{n-1} ,\\
	  & \forall t' \in I: \Delta\pos(t) \vDash \theta_m(t) \text{ and}\\
	 & \forall t' \in I,\ \forall C\in \mathit{cs}:  \spd_{t'}(C)  \text{ and }  \acc_{t'}(C)  \text{ remain in the specified bounds.}
\end{align}

An illustration of the solution space and a solution we are searching for for two cars is depicted in Fig.~\ref{fig:illustration}. As one can see, we want to find out how many splitting points there need to be such that the difference in the position of the cars satisfies the spatial constrains $ \theta_i $ of each phase $ (s_i,I_i) $ of $ m $ as well as the constraints on the speeds of the cars. The possible curves for the relative position that the evolution yields are furthermore constrained by the maximal and minimal acceleration forces possible. Initial values for the relative position and the speeds are fixed, as they are determined by the traffic snapshot from which on we ask for a satisfying sequence of actions. In the figure, we have both a maximum and a minimum spatial constraint on the distance between the cars in each of the phases. Please note that also a single constraint (distance is e.g. greater than some value) or even no constraint (the distance between the cars is not important in this phase) is possible. Despite being possible, the later one should usually not occur because requiring collision freedom should always be included in a specification, which immediately imposes length constraints.

If $ \mathit{DYN}(m,n,I) $ is satisfied, $ n $ splitting points are sufficient to obtain a satisfying sequence of actions that satisfies the behaviour specified by $ m $. We furthermore need a relaxed version $ \mathit{DYN}'(m,n,I) $, which is equivalent to $ \mathit{DYN} $, except that we alter equation $ (4) $ and remove the length constraint on the last phase of $ m $ not reached. Please note that we can rewrite the equations $ (4) $ and $ (5) $ in an equivalent form that does not use quantifiers, so we gain an easy to solve equation system not dealing with quantifiers over infinite domains.

Later, we are interested how $ \mathit{DYN}'$ behaves when answering the question whether or not there is a solution to $ \mathit{DYN} $. For this purpose, we denote with $ \mathit{max\_extension(DYN')} $ the maximal value $ x $ s.t.  $ \mathit{DYN'}(m,n,[0,x]) $ has a solution. Similarly, we denote with $ \mathit{max\_outcome\_pos(DYN')} $ the largest interval $ [y,y'] $ s.t. $ \mathit{DYN}'(m',n,I) $ has a solution, where $ m' $ is identical to $ m $ except that in the last phase, the length constraints are replaced with $ [y,y'] $. $  \mathit{max\_outcome\_spd(DYN')}  $ is the largest interval $ [z,z'] $  for the speed that a car can have when exceeding $ I $ while $ \mathit{DYN}'(m,n,I) $ still has a solution.

Despite not focusing on that topic, we would like to mention that both $ \mathit{max\_outcome\_pos(DYN')} $ and $ \mathit{max\_outcome\_spd(DYN')} $ are vectors, the first one over the differences in positions that are compared in the phases and the second one over the cars.

\begin{figure}
	\begin{center}
		\includegraphics[scale=1.25]{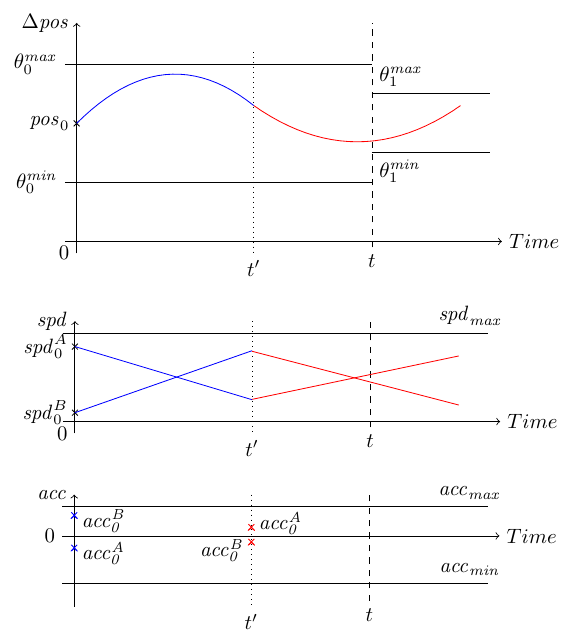}
	\end{center}
	\caption{Dynamic Evolution of two cars on a timing interval with one additional acceleration change for each car in between. For a sequence of actions that satisfies the specification, all values need to stay inside the specified intervals. For the relaxed version $ \mathit{DYN}' $, this does not apply for the last length constraint ($ \theta_1 $ here).}
	\label{fig:illustration}
\end{figure}

Utilising the aforementioned equations, we now can give an answer to the question whether or not there is a satisfying sequence of actions s.t. the cars behave as specified:

\begin{theorem}[Number of Accelerations]
	\label{thm:accChanges}
	Given a traffic snapshot $ \TS $ and timed sequence of states $ m =\langle(s_0,I_0),\dots,(s_n,I_n) \rangle $ with adjacent intervals $ I_i $ and sets $ s_i $ of lengths constraints  $ \theta $ between the cars to achieve, one can decide after finitely many steps whether or not there is a sequence of acceleration changes $ \omega =\langle (\alpha_0,t_0),\dots,(\alpha_{n'},t_{n'}) \rangle $ s.t. $ m(\TS_0,\omega)=  m $.
\end{theorem}

We can decide this question using Alg.~\ref{alg:decidingACC}.

\begin{algorithm}[ht!]
	\caption{Deciding Acceleration}\label{alg:decidingACC}
	\begin{algorithmic}[1]
		\Require Input: sequence of states $ m =\langle p_1,\dots,p_m \rangle$, Interval $ I=[0,t] $\\		
		$ n\gets 0; \ i\gets 1; $		
		\While{$ i\leq m $}	\\
		\ \ \ \ \ \ $ m'\gets \langle p_1,\dots,p_i\rangle; $			
			\While{$ \mathit{max\_extension(DYN'(m',I,n))} \neq \mathit{max\_extension(DYN'(m',I,n+1))}$\\		
				\ \ \ \ \ \ $\lor\mathit{max\_outcome\_pos(DYN'(m',I,n))} \neq		
				 \mathit{max\_outcome\_pos(DYN'(m',I,n+1))}$\\		
				\ \ \ \ \ \ 		$\lor\mathit{max\_outcome\_spd(DYN'(m',I,n))} \neq \mathit{max\_outcome\_spd(DYN'(m',I,n+1))}$}\\			
			\ \ \ \ \ \ \ \ \ \ \ \ $ n\gets n+1 ;$	
			\EndWhile	\\
			\ \ \ \ \ \  $ i\gets i+1 ;$
		\EndWhile		
		\If{$ \mathit{DYN(m,I,n)} $ has a solution}\\		
		\ \ \ \ \ \ \Return $\langle(\acc_{0}(C),t_0),\dots,(\acc_{n}(C),t_{n})\rangle$ of $ \mathit{DYN(m,I,n)} $ (for each car $ C $)
		
		\Else \\		
		\ \ \ \ \ \ \Return no Solution existent.		\EndIf
	\end{algorithmic}
\end{algorithm}

\begin{lemma}[Termination and Correctness]
	\label{lem:alg}
	Algorithm~\ref{alg:decidingACC} terminates \iff a solution is found and returns it or there is no solution at all.
\end{lemma}

\begin{proof}(sketched)
	Alg.~\ref{alg:decidingACC} subsequently maximises the outcome that the dynamics (position and speed) may have after each phase of $ m $, iterating through prefixes $ m' $ of $ m,  $ solving $ \mathit{DYN}' $ for this prefix. The first line of the second while-condition (line 4) ensures that we have sufficiently many splitting points such that we reach the right (time) border of the phase. Line 5 ensures that we maximise the difference in the position between two cars at the end of the current phase $ p_i $, where line 6 maximises the differences in their speeds. If none of these values increases any further within  one iteration, no further iteration will. Aborting then is possible because the outcome of the dynamics is strictly monotone in the number of splitting points and thus has converged against a solution that is maximal for the phase.
\end{proof}

This result is needed in the next section:
\section{Enforcement}
\label{sec:enforcement}

We now present our -- in comparison to \cite{BO23}-- enhanced runtime enforcement approach, utilising the results from the previous section.

In this section, we show how cars can find actions to execute in a distributed manner s.t. the overall evolution of the traffic snapshot satisfies the specified properties up to some time bound. In distinction to previous work \cite{BO23}, where the cars non-deterministically guessed actions to execute, they now only propose timed sequences of states that are valid at least for their own properties. The proposed sequences $ \Pi_C $ are afterwards -- either by one of the cars or by another central entity -- combined into a set of timed sequences of states $ \Pi $, where each sequence represents a combined behaviour of all cars. It is then checked for the existence of a satisfying sequence of actions $ \omega $. If existent, the participating cars get informed over the timed actions $ \omega_c $ they themself must execute to comply to $ \omega $.

As the sequences $ \Pi $ are in the end checked by a single entity, one could argue that it would be easier to refrain from having the specification distributed over all cars. There are, however, several benefits that one gains when using the more distributed approach. First of all, the whole specification does not need to be known beforehand, neither to the other cars nor to the entity that checks if there is a satisfying sequence of actions. Being able to handle such cases is one big strength of runtime enforcement approaches.  Additionally, only the specification up to some time bound needs to be known, not the behaviour beyond that timed horizon, which might not be of interest to the others. Therefore, each car knows its whole specification completely and the central entity/other cars just enough to fulfil its/their task(s). While this argument mostly aimed at privacy concerns, we can also consider it in the light of complexity: The size of the region automaton is exponential in the size of the corresponding SC automaton, which itself is exponential in the size of the specification in SCL. When we considers that only the behaviour up to some time bound is of interest to us, a central SC automaton or even worse, region automaton, would be unnecessary large. 


We now focus on the question wherefrom the cars know which timed sequences of states to announce. This includes getting sequences of regions first (Lemma.~\ref{lemma:sequRegions}) and computing a satisfying timed sequence of states from them (Lemma.~\ref{lemma:SeqTimedStates}).

Before doing so, we would like to mention some results from \cite{BO23}: Given an region automaton $ \mathcal{R}(A_\psi) $ for a specification $ \psi $ in SCL, one can label some of the states as \emph{bad}, these are the once that, if reached, do not allow the run of the region automaton to be extended in a way that allow $ \psi $ to get satisfied. Vice versa, if a sequence ends in a state that is not bad, one can extend it in a way s.t. $ \psi $ is satisfied.

The second case, however, does not hold if we consider specifications $ \varphi $ in TMLSL rather than SCL: Here, it can be the case that the region automaton claims that there is an extension s.t. $ \varphi $ is satisfied, but the cars are not able to behave in a way that conforms to this extension. Thus, the sequences of regions that the region automaton suggests as satisfiable might actually not be satisfiable (but are candidates):

\begin{lemma}[Sequences of regions as potential solutions]
	\label{lemma:sequRegions}
	In every traffic snapshot $ \TS $, one can compute the set $ \Pi $ of sequences of regions $ \pi=\langle[\nu_i],\dots,[\nu_j]\rangle $  that start with the region $ [\nu_i] $ reached in the region automaton $ \mathcal{R}(A_\varphi) $ in the evolution towards $ \TS $ and are candidates for satisfying runs of the region automaton. Moreover, there is no sequence $\pi' $ not in $ \Pi $ but with $ m(\pi')\vDash_t\varphi $.
\end{lemma}

\begin{proof}
	Using Def.~3 of \cite{BO23}, we can compute the set of locations $ \{l_0,\dots,l_n\} $ that the SC automaton $ A_\varphi $ reaches along the evolution towards $ \TS $. Each location $ l\in \{l_0,\dots,l_n\} $ corresponds to a set of regions $ \{[\nu_1],\dots,[\nu_m]\} $, with $ [\nu_i]\vDash \Delta_x(l) $ that is, the region $ [\nu_i] $ satisfies the clock constraints over the history clocks of the location $ l $ and especially $ p\in\mathcal{L}(l_i)\ \iff\ [\nu_i](p)=0 $ for every proposition $ p $. We ignore constraints over the prophecy clocks here, because the future is (at least at the end of the sequence) unknown and the history clocks are sufficient for determining the intervals.
\end{proof}

Given a sequence of regions, we can compute a timed sequence of states that satisfies the sequence of regions:

\begin{lemma}[From regions to timed sequences of states]
	\label{lemma:SeqTimedStates}
	For every sequence of regions $ \pi=\langle[\nu_0],\dots,[\nu_n]\rangle  $ one can construct a timed sequence of state $ m(\pi)=\langle(s_0,I_0),\dots,(s_m,I_m) \rangle $ s.t. $ m\vDash\pi $.
\end{lemma}

\begin{proof}
	For simplicity, we assume that there is a global clock that is not reset, counting the time from the beginning of the sequence. Starting with $ [\nu_0] $ and subsequently going trough all $ [\nu_i] $, we determine for each point in time $ t $ which propositions $ p $ are valid in it, which is achieved by looking at formulae of the form $ x_p=0 $. To determine the shape of the intervals ($ [],(),[)\text{ or }(] $), we consider the (in-)equalities in the regions: If some $ p $ is valid in the next point in time, $ y_p=1 $ leads to a closed interval border (\enquote{]}),  $ y_p<1 $ leads to an open one (\enquote{)}). We do the same for the history clocks $ x_p $.
\end{proof}

As all cars announce timed sequences of states, we need to combine them into a single sequence that the central entity can check:

\begin{lemma}[Combining Sequences of states]
	\label{lem:combining}
	Given two finite timed sequences of states $ m_1 $ and $ m_2 $, one can construct a timed sequence $ m $ s.t. for every $ \varphi:\ m_1\vDash\varphi \lor m_2\vDash \varphi \implies m\vDash\varphi$.
\end{lemma}

\begin{proof}
	We start with an \enquote{empty} sequence $ m=\langle(\_,[0,0]),(\_(0,1)),(\_,[1,1]),(\_,(1,2)), \dots(\_,I_n) \rangle $. Going through each state $ (\_,I_i) $ of $ m $, we look in both $ m_1 $ and $ m_2 $ and insert the propositions from the states $ (s_j,I_j) $, where $ I_j $ contains $ I_i $. If it happens that for some two neighbouring states $ (s_i,I_i) $ and $ (s_{i+1},I_{i+1}), \ s_{i}=s_{i+1}$ holds, we can fuse the two into a single state $ (s_i, I_i+I_{i+1}) $.	As a last step, we check if the resulting sequence is consistent. If it happens that there is a contradiction in one of the states, say $ \cl{A} $ and $ \lnot \cl{A} $ need to hold in the same (time) interval, the timed sequence of states cannot be satisfied at all and is thus invalid. 
\end{proof}

Please note that the other direction not necessarily holds, as e.g. $ p1\land p_2 $ could hold in $ [1,1] $ of $ m $, but in $ m_1 $ only $ p_1 $ and in $ m_2 $ only $ p_2 $ holds in the respective interval.

We utilise the aforementioned results in the controllers of the cars and the central entity that determines whether a solution exists. The controller is depicted in Fig.~\ref{fig:controller} and the central entity $\mathit{RSU}$ (Road-Side Unit) in Fig.~\ref{fig:centralEntity}.

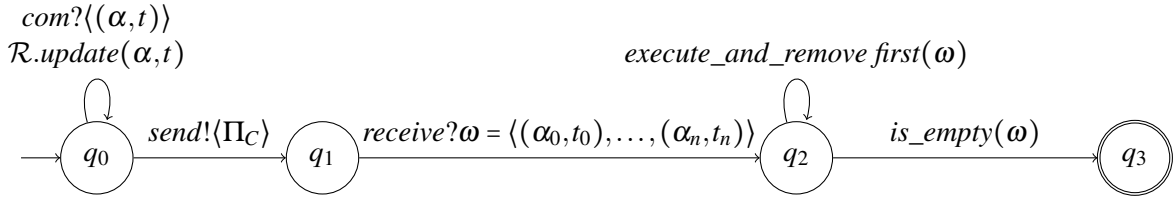
\begin{figure}[h!]
	\begin{tikzpicture}[node distance =45mm,initial text=,]
		\node[state,initial] (n0) {$ q_0 $};
		\node[state,right of= n0,xshift=-15mm] (n1) {$ q_1 $};
		\node[state,right of = n1,xshift=18mm] (n2) {$ q_2 $};
		\node[state,right of = n2,xshift=0mm,accepting] (n3) {$ q_3 $};

		\path[->]
		(n0) edge[loop above] node[midway,above,align=center]{$ \mathit{com}?\langle(\alpha,t)\rangle $\\ $ \mathcal{R}.\mathit{update}(\alpha,t) $}(n0)
		(n0) edge node[midway,above,align=center]{$ \mathit{send!}\langle\Pi_C\rangle $}(n1)
		(n1) edge node [midway,above]{$ \mathit{receive}? \omega= \langle (\alpha_0,t_0),\dots,(\alpha_n,t_n)\rangle $}(n2)
		(n2) edge[loop above] node[midway,above]{$ \mathit{execute\_and\_remove\ first}(\omega) $}(n2)
		(n2) edge node[midway,above]{$ \mathit{is\_empty}(\omega) $}(n3)
		;
	\end{tikzpicture}
	\caption{Controller for each car $ C $, each is equipped with an instance of it. The controller keeps track of the traffic situation in $ q_0 $ and updates the sequences it could announce accordingly. When announcing the sequences $ \Pi_C $, it proceeds to $ q_1 $ and waits for a positive response from the central entity and executes the sequence of actions $ \omega $ that it received from there, until there are not further actions to execute. Please note that we omitted clock constraints to actually force the controller to leave a state.}
	\label{fig:controller}
\end{figure}

Both of them use several functions on their transitions. In the controller, $ \mathcal{R}.\mathit{update}(\alpha,t) $ is used so that the internal region automaton keeps track about the behaviour on the road and thus is in the correct state(s), before the enforcement mechanism is triggered. In this location, we may already have an evolution that leads to a state s.t. all further extensions are unsatisfiable. If such a behaviour is undesired, constraints should be added s.t. one does not stay in this location.  $ \mathit{execute\_and\_remove\_first}(\omega) $ takes the first time stamped action $ (\alpha_1,t_1) $ from the action sequence $ \omega=\langle(\alpha_1,t_1),\dots,(\alpha_n,t_n)\rangle $, waits until the clock reaches $ t_1 $ and executes $ \alpha_1 $. Afterwards, this element is removed from $ \omega $, so that the next action is ready to be executed. $ \mathit{is\_empty}(\omega) $ is true for the empty sequence $ \langle\rangle $.

In $ \mathit{RSU} $, the function $ \mathit{D.push}(\Pi_C) $ is used to internally store the announced sets of timed sequences of states $ \Pi_C $ for each car $ C $ in some data structure $ D $. The set of sequences representing all possible satisfying sequences for all cars is constructed using $ \mathit{combine}(C_1,\dots,C_n) $ and stored in $ \Pi $. Using Alg.~\ref{alg:decidingACC}, it can then decide whether or not one of the sequences in $ \Pi $ is one for which there is a satisfying sequence of actions. If so, the solution $ \omega $ is computed and afterwards split into single solutions $ \omega_c $ for each car $ c $, so that every car only gets informed of the actions it itself has to execute. After informing a car, it is removed from the data structure $ D $.

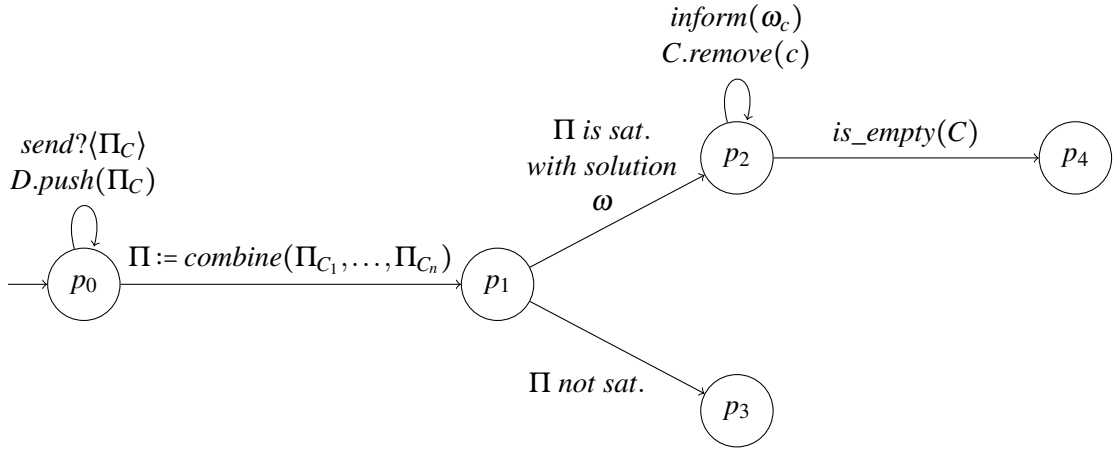
\begin{figure}
	\begin{tikzpicture}[node distance =45mm,initial text=,]
		\node[state,initial] (n0) {$ p_0 $};
		\node[state,right of=n0,xshift=10mm] (n1) {$ p_1 $};
		\node[state, above right of = n1,yshift=-15mm] (n2) {$ p_2 $};
		\node[state, below right of = n1,yshift=15mm] (n3) {$ p_3 $};
		\node[state, right of= n2] (n4) {$ p_4 $};
		
		\path[->]
		(n0) edge[loop above] node[midway,above,align=center]{$ \mathit{send}?\langle\Pi_C\rangle $\\$ \mathit{D.push}(\Pi_C) $}(n0)
		(n0) edge[] node[midway,above,align=center]{$ \Pi:=\mathit{combine}(\Pi_{C_1},\dots,\Pi_{C_n}) $}(n1)
		(n1) edge[] node[above,midway,align=center,xshift=-2mm]{$ \Pi \mathit{\ is\ sat.}$ \\ $\mathit{with\ solution}$\\ $ \omega$}(n2)
		(n1) edge[] node[midway, below,xshift=-4mm,yshift=-2mm]{$ \Pi\ \mathit{not\ sat.} $}(n3)
		(n2) edge[loop above] node[midway,above,align=center]{$\mathit{inform}(\omega_c)$\\$ C.\mathit{remove(c)} $} (n2)
		(n2) edge[] node[midway, above]{$\mathit{is\_empty}(C)$}(n4)
		;
	\end{tikzpicture}
\caption{Central Entity $ \mathit{RSU} $. $ \mathit{RSU} $ waits in $ p_0 $ for announced sequences $ \Pi_{C_i} $ from the cars and combines them into a single set of sequences $ \Pi $, which it than can check for a satisfying run. If positively answered, it sends the sequence to execute to each car.}
\label{fig:centralEntity}
\end{figure}

The communication takes place over the channels $ \mathit{com} $, over which the cars announce action they execute before the enforcement mechanism is triggered. Channel $ \mathit{send} $ is used to inform $ \mathit{RSU}$ about the possible plans $ \Pi_c $ of each car. $ \mathit{receive} $ is used for the opposite direction, informing the cars which actions to execute.

\begin{theorem}[Correct- and Soundness]
	If the controller from Fig.~\ref{fig:controller} proceeds to location $ q_3 $, the specification of all cars is satisfied up to the given time bound $ t $. If it cannot proceed to location $ q_3 $, then there is no sequence of actions for the cars to take that respects the specification of all cars.
\end{theorem}

\begin{proof}
	Due to Lemma~\ref{lemma:sequRegions}, we get all possible satisfying sequences of regions of cars. Due to Lemma~\ref{lemma:SeqTimedStates}, we can compute equivalent timed sequences of states from them. Lemma~\ref{lem:combining} lets us combine them towards some $ \Pi $ on the transition from $ p_0 $ to $ p_1 $ in $ \mathit{RSU} $ s.t. all possible combinations of solutions for all cars are considered. Each of them is checked using Alg.~\ref{alg:decidingACC}, so the due to Lemma~\ref{lem:alg}, the solution found is a correct one.
\end{proof}


If $ \mathit{RSU} $ reaches location $ p_3 $, the specification is unsatisfiable, so there is no sequence of actions for the cars to execute. In this case collision freedom can still be guaranteed (assuming that there were no collisions yet), as the reservations of the cars occupy a space big enough to come to a standstill within that space.

\begin{remark}[Number of sequences to consider]
	If all cars announce all timed sequences of states that satisfy their specification, the central entity/road-side unit needs to check all combinations of these sequences (with exactly one sequence in the combination for each car), resulting in a lot of computation. However, these computations do not depend on each other and can thus be parallelised. If we consider that the cars themself compute this, rather than some road cite unit, one can think of a more advanced protocol than the one proposed here, where the cars distribute the sequences to check and thus the computational effort between each other.
\end{remark}

\begin{remark}[Discrete Actions]
	Through both Sect.~\ref{sec:decidability} and Sect.~\ref{sec:enforcement}, we only considered how the cars can change their accelerations to ensure that they satisfy the length measurements in the specification, ignoring the discrete actions completely. As said, given a timed sequence of states, it is easy to see what discrete actions are to execute when, as they directly change the formulae valid and thus the phase. Some of them, however, need to be respected when constructing the length comparisons $ \theta_i $ that we check in $ \mathit{DYN} $.
\end{remark}

\section{Conclusion}

\paragraph{Contribution}
In this paper, we proposed a runtime enforcement approach for autonomous car in motorway traffic, employing communication between the cars, where the knowledge about the satisfaction of a property is represented using a region automaton. In answering the question whether or not a specification (now in the form of a timed sequence of states) is satisfiable, we were able to eliminate one of the roots for the semi-decidability of the satisfiability problem of TMLSL.

\paragraph{Future Work}
Future work on the topic includes studying the satisfiability problem of TMLSL again, in an effort to show that the logic is indeed decidable over infinite runs. Further topics also include the extension of the logic and the proposed runtime enforcement approach towards the aforementioned more complex road topologies. Both of them offer some challenges in the semantics and runtime enforcement, as their models are more complicated that the ones for motorway traffic. For urban traffic, the assumption that there is a central entity that all cars can communicate with is not too far from reality, as on almost all intersections traffic lights are present, some of which already communicate with the buses that cross/approach them.

Steps towards an implementation for solving the decidability problem of $ TMLSL $ were made and could be adjusted to be used in the runtime enforcement setting. With an implementation, we could also examine if the proposed approach is suited for real-time applications like car control on motorways, e.g. the computation happens fast enough.

\paragraph{Acknowledgements.}
We thank the anonymous reviewers for their helpful comments.

\bibliographystyle{eptcs}
\bibliography{./bib.bib}
\end{document}